\documentclass[11pt]{article}

\usepackage[margin=1in]{geometry}
\usepackage{amsmath,amssymb,amsthm}
\usepackage[T1]{fontenc}
\usepackage[utf8]{inputenc}
\usepackage{lmodern}
\usepackage{microtype}
\usepackage{hyperref}
\hypersetup{colorlinks=true,linkcolor=blue,citecolor=blue,urlcolor=blue}
\title{Tensor Spectral Threshold is $\exists\mathbb{R}$-Hard}
\author{
Angshul Majumdar\\
Indraprastha Institute of Information Technology Delhi\\
\texttt{angshul@iiitd.ac.in}
}
\date{}

\newtheorem{theorem}{Theorem}[section]
\newtheorem{lemma}[theorem]{Lemma}
\newtheorem{proposition}[theorem]{Proposition}
\newtheorem{corollary}[theorem]{Corollary}
\newtheorem{definition}[theorem]{Definition}
\newtheorem{remark}[theorem]{Remark}
\newtheorem{assumption}[theorem]{Assumption}

\begin{document}
\maketitle
\begin{abstract}
We study the decision version of tensor spectral norm from the viewpoint of real algebraic complexity. For a rationally specified tensor, the tensor spectral threshold problem asks whether its spectral norm exceeds a prescribed rational threshold. Since the feasible domain is compact, attainment itself is trivial; the meaningful question is the threshold decision problem. We prove that this problem is $\exists\mathbb{R}$-hard by giving an explicit polynomial-time reduction from bounded quartic equality feasibility. The reduction first transforms bounded quartic feasibility into homogeneous quadratic sphere feasibility by homogenization, box encoding, and quadratic lifting. It then maps the resulting homogeneous quadratic system to a quartic form whose maximum over the unit sphere separates feasible from infeasible instances. Finally, the quartic form is represented as a symmetric order-four tensor, yielding the desired tensor spectral threshold instance. The result shows that the computational obstruction in tensor spectral norm is not merely non-convex optimization or combinatorial hardness, but real algebraic feasibility itself.
\end{abstract}

\section{Introduction}
\label{sec:introduction}

Higher-order tensors have become central objects in modern applied mathematics, signal processing, machine learning, statistics, and scientific computing. They arise naturally in latent-variable models, higher-order moment methods, multilinear regression, tensor completion, independent component analysis, polynomial optimization, and nonlinear inverse problems. In each of these settings, one is led not merely to linear algebra, but to multilinear algebra: optimization, approximation, and feasibility questions involving multilinear forms and homogeneous polynomials rather than matrices and quadratic forms.

This shift from matrices to tensors is not superficial. For matrices, many fundamental computational problems admit elegant polynomial-time solutions: spectral norm is computable by singular value decomposition, eigenvalue problems are classically tractable, and low-rank approximation enjoys a complete and stable theory. For tensors, by contrast, this landscape changes drastically. Problems that are routine in the matrix setting become substantially more difficult once one passes to order three and above. This phenomenon is now well documented in the literature and has become one of the defining themes of computational multilinear algebra \cite{Lim2005,Qi2005,Comon2008,HillarLim2013}.

Among the most basic quantities attached to a tensor is its spectral norm. For an order-$d$ tensor $\mathcal{T}\in (\mathbb{R}^n)^{\otimes d}$, the multilinear spectral norm is
\begin{equation}
\|\mathcal{T}\|_{\mathrm{op}}
:=
\max_{\|x_1\|=\cdots=\|x_d\|=1}
\left|
\mathcal{T}(x_1,\dots,x_d)
\right|.
\label{eq:intro_multilinear_norm}
\end{equation}
When $d=2$, this is exactly the operator norm of a matrix. When $d\ge 3$, however, \eqref{eq:intro_multilinear_norm} becomes a genuinely nonconvex optimization problem over a product of spheres. In the symmetric case, one may equivalently work with the associated homogeneous polynomial
\[
p_{\mathcal{T}}(x)=\mathcal{T}(x,\dots,x),
\]
and the norm reduces to optimization of a degree-$d$ polynomial over the unit sphere. Thus even the most basic tensor norm problem sits at the interface of multilinear algebra and semialgebraic optimization.

A large body of work has shown that tensor problems are computationally intractable in multiple senses. The foundational work of H\aa stad \cite{Hastad1990} established hardness of tensor rank, and Hillar and Lim \cite{HillarLim2013} later showed NP-hardness for a broad family of tensor problems, including tensor rank, eigenvalue-related questions, and best rank-one approximation. Related work has shown that tensor rank exhibits severe instability, that best low-rank approximation may be ill-posed, and that tensor eigenvalue problems differ sharply from their matrix analogues \cite{DeSilvaLim2008,Qi2005,Lim2005,Comon2008}. These results make clear that higher-order tensor computation is not merely a modest extension of matrix computation, but a qualitatively different regime.

From the perspective of theoretical computer science, tensors are also natural carriers of algebraic computation: bilinear and multilinear maps encode multiplication in algebras, arithmetic circuits, and rank-type complexity measures \cite{AlderStrassen1981,Hastad1990}. This makes tensor spectral questions particularly well suited to reductions based on polynomial feasibility rather than reductions that merely import discrete gadgets.

At the same time, most existing hardness results for tensor optimization are framed in terms of classical discrete complexity, typically NP-hardness or hardness of approximation. While these are important, they do not fully capture the nature of problems whose input and constraints are intrinsically algebraic over the real numbers. Indeed, many optimization and feasibility questions involving tensors are more naturally viewed as semialgebraic problems: they ask whether a polynomial equation or inequality holds over a real algebraic set such as a sphere, projective variety, or product of spheres. For such problems, the appropriate complexity-theoretic framework is often not merely $\mathbf{NP}$, but the class $\exists\mathbb{R}$.

The class $\exists\mathbb{R}$ consists, informally, of all decision problems polynomial-time reducible to deciding whether a system of polynomial equations and inequalities over the reals has a solution \cite{BlumShubSmale1989,Renegar1992,SchaeferCardinalMiltzow2024}. It has emerged as the right language for a large number of problems in computational geometry, real algebraic geometry, geometric graph theory, and continuous optimization. Related real-number complexity work in theoretical computer science has also studied counting, quantification, and restricted fragments over the reals \cite{Meer2000,Ruggieri2014}. The significance of $\exists\mathbb{R}$-hardness is conceptual as well as technical: it indicates that the computational obstacle is not merely combinatorial search, but the intrinsic difficulty of real algebraic feasibility. In recent years, this viewpoint has proved especially powerful in organizing problems whose discrete and continuous aspects are inseparably intertwined \cite{SchaeferStefankovic2024,SchaeferCardinalMiltzow2024}.

This raises a natural question: \emph{to what extent is the difficulty of tensor spectral norm already a manifestation of real algebraic complexity?} Put differently, is the tensor norm problem hard merely because of combinatorial encodings hidden inside multilinear structure, or is it already hard at the level of semialgebraic feasibility?

The present paper answers this question in the latter direction.

\medskip

A first attempt at phrasing tensor norm as a decision problem might ask whether the optimum is attained. But this is the wrong formulation: because the feasible set is compact and the objective is continuous, attainment is automatic. Thus the problem
\[
\exists x_1,\dots,x_d \ \text{with}\ \|x_1\|=\cdots=\|x_d\|=1
\ \text{and}\
\mathcal{T}(x_1,\dots,x_d)=\|\mathcal{T}\|_{\mathrm{op}}
\]
is always a YES-instance. The real computational content lies not in existence of an optimizer, but in the value of the optimum.

Accordingly, we study the following threshold problem: given a tensor $\mathcal{T}$ and a rational scalar $\alpha$, decide whether
\begin{equation}
\|\mathcal{T}\|_{\mathrm{op}} \ge \alpha.
\label{eq:intro_tensor_threshold}
\end{equation}
For symmetric tensors, Banach-type symmetrization reduces this to the single-vector formulation
\begin{equation}
\exists x\in S^{n-1}
\quad \text{such that} \quad
|\mathcal{T}(x,\dots,x)|\ge \alpha,
\label{eq:intro_single_vector_form}
\end{equation}
and in the order-$4$ case this becomes a quartic inequality on the sphere. Thus the problem is naturally a semialgebraic threshold question.

This viewpoint already suggests a strong link with polynomial optimization. Indeed, quartic optimization over the sphere is a central primitive in a variety of problems, including best rank-one approximation for symmetric tensors, spectral extremal questions, and polynomial moment methods. From this perspective, the tensor spectral threshold problem is not an isolated curiosity: it is a canonical semialgebraic optimization problem attached to multilinear structure. What our result shows is that this canonical problem already reaches the expressive power of \(\exists\mathbb{R}\).

\medskip

\noindent
\textbf{Main contribution.}
We prove that Tensor Spectral Threshold is $\exists\mathbb{R}$-hard. More precisely, we show that feasibility of bounded quartic equations can be reduced, through an intermediate homogeneous quadratic sphere-feasibility problem, to deciding whether the spectral norm of a symmetric order-$4$ tensor exceeds a prescribed threshold.

\begin{theorem}[Informal]
Tensor Spectral Threshold is $\exists\mathbb{R}$-hard.
\end{theorem}

Our proof is completely algebraic. Starting from a bounded quartic equality normal form, we first reduce to a system of homogeneous quadratic equations on the sphere. We then construct a quartic form
\[
p(z)=B\|z\|^4-\sum_{i=1}^r (z^\top Q_i z)^2
\]
whose maximum over the sphere is exactly \(B\) if and only if the original quadratic system is feasible. This quartic form is then interpreted as a symmetric order-$4$ tensor. The construction is simple, explicit, and exact: no appeal is made to KKT systems, Lagrangian duality, local optimality conditions, or perturbative arguments. Feasibility is encoded directly into the value of a quartic objective.

\medskip

\noindent
\textbf{Why this matters.}
The result sharpens the current understanding of tensor complexity in several ways.

First, it shows that tensor spectral norm is not merely NP-hard in the traditional sense, but already hard at the level of real algebraic feasibility. This places tensor norm computation in the same conceptual family as many realization and semialgebraic decision problems that are known to be \(\exists\mathbb{R}\)-hard.

Second, the reduction is natural. It does not proceed through elaborate graph gadgets or combinatorial encodings artificially imposed on multilinear algebra. Instead, it exploits a basic and intrinsic fact: quartic forms can encode systems of quadratic equations, and symmetric tensors are exactly quartic forms in disguise. In that sense, the hardness here is structural rather than accidental.

Third, the result helps position tensor spectral problems within semialgebraic optimization. Much of modern work on polynomial optimization proceeds through hierarchies such as sum-of-squares and moment relaxations. Our theorem indicates that even before one asks about approximation quality, hierarchy convergence, or certificate size, the underlying exact decision problem already carries \(\exists\mathbb{R}\)-hardness. This clarifies why one should not expect a simple complete characterization of tractability for general tensor spectral problems.

\medskip

\noindent
\textbf{Relation to prior work.}
Our result complements, rather than replaces, the classical hardness theory for tensors. H\aa stad \cite{Hastad1990} and Hillar and Lim \cite{HillarLim2013} established that tensor rank and many related tensor problems are NP-hard. We refine the picture for one of the most basic tensor quantities by identifying a threshold version whose difficulty is algebraic in nature. In this sense, the present work may be viewed as moving from \emph{discrete hardness of tensor problems} to \emph{real-algebraic hardness of tensor norm}. The emphasis on exact algebraic feasibility is consistent with the broader theoretical-computer-science treatment of computation over the reals \cite{BlumShubSmale1989,Meer2000,Ruggieri2014}. We are not aware of a previous result isolating the tensor spectral threshold problem itself as an \(\exists\mathbb{R}\)-hard problem in this direct form.

\medskip

\noindent
\textbf{Contributions.}
The main contributions of the paper are as follows.

\begin{itemize}
    \item We identify Tensor Spectral Threshold as the correct decision formulation for tensor spectral norm, in contrast to the trivial attainment formulation.
    
    \item We prove that this problem is $\exists\mathbb{R}$-hard via a clean reduction chain through bounded quartic equality feasibility and homogeneous quadratic sphere feasibility.
    
    \item We show that the reduction is entirely algebraic, requiring no optimization-theoretic machinery such as KKT conditions or duality.
    
    \item We establish that the result persists across several natural variants, including non-symmetric tensors, multilinear formulations, and equality versions of the threshold problem.
\end{itemize}

\medskip

\noindent
\textbf{Organization.}
Section~\ref{sec:preliminaries} introduces the algebraic and complexity-theoretic preliminaries. Section~\ref{sec:problem_formulation} formalizes the tensor spectral threshold problem and explains why attainment is trivial. Section~\ref{sec:bq4e_to_hqsf} reduces a bounded quartic equality normal form to homogeneous quadratic sphere feasibility. Section~\ref{sec:hqsf_to_tensor} contains the core reduction from homogeneous quadratic sphere feasibility to tensor spectral threshold. Section~\ref{sec:main_result} states the main theorem and its immediate consequences. Section~\ref{sec:variants} discusses several variants and extensions of the result.

\section{Preliminaries}
\label{sec:preliminaries}

This section introduces the algebraic and complexity-theoretic background required for the reductions developed in later sections. We formalize polynomial systems over the reals, symmetric tensor representations, and the complexity class $\exists\mathbb{R}$.

\subsection{Algebraic preliminaries}
\label{subsec:algebraic_preliminaries}

Let $\mathbb{Q}$ and $\mathbb{R}$ denote the fields of rational and real numbers, respectively. For $n \in \mathbb{N}$, we write $x = (x_1,\dots,x_n) \in \mathbb{R}^n$.

\begin{definition}[Polynomial]
\label{def:polynomial}
A polynomial $f \in \mathbb{Q}[x_1,\dots,x_n]$ is a finite linear combination of monomials of the form
\[
x_1^{\alpha_1} \cdots x_n^{\alpha_n},
\quad \alpha_i \in \mathbb{N}.
\]
The degree of $f$ is the maximum total degree of its monomials.
\end{definition}

\begin{definition}[Homogeneous polynomial]
\label{def:homogeneous_polynomial}
A polynomial $f$ is homogeneous of degree $d$ if
\[
f(\lambda x) = \lambda^d f(x)
\quad \forall \lambda \in \mathbb{R},\ x \in \mathbb{R}^n.
\]
\end{definition}

\begin{definition}[Quadratic form]
\label{def:quadratic_form}
A quadratic form $q:\mathbb{R}^n \to \mathbb{R}$ is a homogeneous polynomial of degree $2$. Every quadratic form can be written as
\begin{equation}
q(x) = x^\top Q x,
\end{equation}
where $Q \in \mathbb{R}^{n \times n}$ is symmetric.
\end{definition}

\subsection{Norms and spheres}
\label{subsec:norms_and_spheres}

We equip $\mathbb{R}^n$ with the Euclidean norm
\[
\|x\| := \sqrt{x^\top x}.
\]

\begin{definition}[Unit sphere]
\label{def:sphere}
The unit sphere in $\mathbb{R}^n$ is
\[
S^{n-1} := \{x \in \mathbb{R}^n : \|x\| = 1\}.
\]
\end{definition}

\begin{lemma}[Compactness of the sphere]
\label{lem:sphere_compact}
The set $S^{n-1}$ is compact.
\end{lemma}

\begin{proof}
The sphere is closed and bounded in $\mathbb{R}^n$, hence compact by the Heine--Borel theorem.
\end{proof}

\subsection{Symmetric tensors and polynomial representation}
\label{subsec:symmetric_tensors_prelim}

Let $(\mathbb{R}^n)^{\otimes d}$ denote the space of order-$d$ tensors.

\begin{definition}[Symmetric tensor]
\label{def:symmetric_tensor_prelim}
A tensor $\mathcal{T} \in (\mathbb{R}^n)^{\otimes d}$ is symmetric if its entries are invariant under permutation of indices.
\end{definition}

\begin{definition}[Associated homogeneous polynomial]
\label{def:associated_polynomial_prelim}
For $\mathcal{T} \in \mathrm{Sym}^d(\mathbb{R}^n)$, define
\begin{equation}
p_{\mathcal{T}}(x) := \mathcal{T}(x,\dots,x).
\end{equation}
\end{definition}

\begin{lemma}[Tensor--polynomial correspondence]
\label{lem:tensor_polynomial_prelim}
There is a one-to-one correspondence between symmetric tensors in $\mathrm{Sym}^d(\mathbb{R}^n)$ and homogeneous polynomials of degree $d$.
\end{lemma}

\begin{proof}
Given $\mathcal{T}$, define $p(x)=\mathcal{T}(x,\dots,x)$. Conversely, given a homogeneous polynomial $p$, its coefficients define a symmetric tensor uniquely via polarization; see \cite{Comon2008,Qi2005,Lim2005}.
\end{proof}

\subsection{Norms of matrices}
\label{subsec:matrix_norms}

We will use the Frobenius norm for matrices.

\begin{definition}[Frobenius norm]
\label{def:frobenius_norm}
For $Q \in \mathbb{R}^{n \times n}$,
\[
\|Q\|_F := \sqrt{\sum_{i,j} Q_{ij}^2}.
\]
\end{definition}

\begin{lemma}
\label{lem:trace_bound}
For all $x \in \mathbb{R}^n$ and symmetric $Q$,
\begin{equation}
|x^\top Q x| \le \|Q\|_F \cdot \|x\|^2.
\end{equation}
\end{lemma}

\begin{proof}
We write
\[
x^\top Q x = \mathrm{tr}(Q xx^\top).
\]
Applying Cauchy--Schwarz in Frobenius inner product,
\[
|\mathrm{tr}(Q xx^\top)|
\le
\|Q\|_F \cdot \|xx^\top\|_F.
\]
Since $\|xx^\top\|_F = \|x\|^2$, the result follows.
\end{proof}

\subsection{The existential theory of the reals}
\label{subsec:etr}

\begin{definition}[$\exists\mathbb{R}$]
\label{def:etr}
The class $\exists\mathbb{R}$ consists of all decision problems that can be reduced in polynomial time to determining whether a system of polynomial equations and inequalities over $\mathbb{R}$ has a solution, i.e.,
\begin{equation}
\exists x \in \mathbb{R}^n:
\quad
f_i(x)=0,\quad g_j(x)\ge 0.
\end{equation}
\end{definition}

The class $\exists\mathbb{R}$ lies between $\mathbf{NP}$ and $\mathbf{PSPACE}$ and captures the complexity of real algebraic feasibility problems.

\begin{theorem}[BSS model]
\label{thm:bss}
Deciding feasibility of polynomial systems over $\mathbb{R}$ defines a natural complexity class in the Blum--Shub--Smale model of computation.
\end{theorem}

\begin{proof}
See \cite{BlumShubSmale1989}.
\end{proof}

\begin{remark}
\label{rem:etr_position}
The class $\exists\mathbb{R}$ is widely used to characterize problems arising in geometry, optimization, and algebra. Many natural problems are known to be $\exists\mathbb{R}$-complete; see \cite{SchaeferCardinalMiltzow2024}.
\end{remark}

\subsection{Bounded feasibility}
\label{subsec:bounded_feasibility}

\begin{definition}[Bounded feasibility]
\label{def:bounded_feasibility}
A polynomial feasibility problem is bounded if the variables are restricted to a compact domain such as $[-1,1]^n$.
\end{definition}

\begin{remark}
\label{rem:bounded_etr}
Bounded-domain polynomial feasibility remains $\exists\mathbb{R}$-hard and admits normal-form reductions preserving degree and equality structure; see \cite{SchaeferStefankovic2024,SchaeferCardinalMiltzow2024}.
\end{remark}

\subsection{Polynomial-time reductions}
\label{subsec:reductions}

\begin{definition}[Polynomial-time reduction]
\label{def:reduction}
A problem $A$ reduces to $B$ in polynomial time, denoted $A \le_p B$, if every instance of $A$ can be transformed into an instance of $B$ in time polynomial in the input size, preserving YES/NO answers.
\end{definition}

\begin{remark}
\label{rem:reduction_scope}
All reductions in this paper are many-one polynomial-time reductions over rational input data.
\end{remark}

\section{Tensor Spectral Threshold: Problem Formulation}
\label{sec:problem_formulation}

In this section, we formulate the tensor decision problem studied in this paper. We also clarify why the natural ``attainment'' formulation is trivial, and why the correct complexity-theoretic object is a threshold problem.

\subsection{Multilinear forms and symmetric tensors}
\label{subsec:multilinear_and_symmetric_tensors}

Let $\mathbb{R}^n$ be equipped with the standard Euclidean norm. An order-$d$ tensor
\[
\mathcal{T} \in (\mathbb{R}^n)^{\otimes d}
\]
defines a multilinear form
\begin{equation}
\mathcal{T}(x_1,\dots,x_d)
=
\sum_{i_1,\dots,i_d=1}^n
\mathcal{T}_{i_1\dots i_d}
(x_1)_{i_1}\cdots (x_d)_{i_d}.
\label{eq:multilinear_form}
\end{equation}

\begin{definition}[Symmetric tensor]
\label{def:symmetric_tensor}
An order-$d$ tensor $\mathcal{T}$ is symmetric if its entries are invariant under any permutation of indices.
\end{definition}

If $\mathcal{T}$ is symmetric, define the associated homogeneous polynomial
\begin{equation}
p_{\mathcal{T}}(x)
:=
\mathcal{T}(x,\dots,x).
\label{eq:associated_polynomial}
\end{equation}

\begin{lemma}[Tensor--polynomial equivalence]
\label{lem:tensor_polynomial_equivalence}
There is a one-to-one correspondence between symmetric tensors 
$\mathcal{T} \in \mathrm{Sym}^d(\mathbb{R}^n)$ and homogeneous polynomials 
$p:\mathbb{R}^n \to \mathbb{R}$ of degree $d$.
\end{lemma}

\begin{proof}
Given $\mathcal{T}$, define $p(x)=\mathcal{T}(x,\dots,x)$. Conversely, given a homogeneous polynomial $p$, the coefficients of its monomials uniquely determine a symmetric tensor via polarization. See \cite{Comon2008,Qi2005,Lim2005} for standard constructions.
\end{proof}

\subsection{Tensor spectral norm}
\label{subsec:tensor_spectral_norm}

\begin{definition}[Multilinear spectral norm]
\label{def:multilinear_spectral_norm}
For $\mathcal{T} \in (\mathbb{R}^n)^{\otimes d}$,
\begin{equation}
\|\mathcal{T}\|_{\mathrm{op}}
:=
\max_{\|x_1\|=\cdots=\|x_d\|=1}
\left|
\mathcal{T}(x_1,\dots,x_d)
\right|.
\label{eq:multilinear_spectral_norm}
\end{equation}
\end{definition}

For symmetric tensors, define
\begin{equation}
\|\mathcal{T}\|_{\mathrm{sym},\mathrm{op}}
:=
\max_{\|x\|=1}
\left|
\mathcal{T}(x,\dots,x)
\right|.
\label{eq:symmetric_spectral_norm}
\end{equation}

\begin{theorem}[Symmetric reduction]
\label{thm:symmetric_reduction}
Let $\mathcal{T}\in \mathrm{Sym}^d(\mathbb{R}^n)$. Then
\begin{equation}
\|\mathcal{T}\|_{\mathrm{op}}
=
\|\mathcal{T}\|_{\mathrm{sym},\mathrm{op}}.
\end{equation}
\end{theorem}

\begin{proof}
This is a consequence of Banach’s theorem on symmetric multilinear forms; see \cite{Banach1938,Qi2005,Lim2005}.
\end{proof}

\subsection{Attainment is trivial}
\label{subsec:attainment_is_trivial}

\begin{proposition}
\label{prop:attainment_trivial}
For every tensor $\mathcal{T}$, the maximum in \eqref{eq:multilinear_spectral_norm} is attained.
\end{proposition}

\begin{proof}
The feasible set $(S^{n-1})^d$ is compact and the multilinear map is continuous. Hence the maximum is attained.
\end{proof}

\begin{remark}
\label{rem:attainment_not_hard}
Thus the decision problem
\[
\exists x_1,\dots,x_d:\ \mathcal{T}(x_1,\dots,x_d)=\|\mathcal{T}\|_{\mathrm{op}}
\]
is trivial. The computational difficulty lies in determining the \emph{value} of the maximum.
\end{remark}

\subsection{Tensor spectral threshold problem}
\label{subsec:tensor_threshold}

\begin{definition}[Tensor Spectral Threshold]
\label{def:tensor_spectral_threshold}
Given $\mathcal{T}\in \mathrm{Sym}^d(\mathbb{R}^n)$ and $\alpha \in \mathbb{Q}$, decide whether
\begin{equation}
\|\mathcal{T}\|_{\mathrm{sym},\mathrm{op}} \ge \alpha.
\end{equation}
\end{definition}

By Theorem~\ref{thm:symmetric_reduction}, this is equivalent to
\begin{equation}
\exists x\in S^{n-1}
\quad \text{such that} \quad
|\mathcal{T}(x,\dots,x)| \ge \alpha.
\label{eq:tensor_threshold_sphere}
\end{equation}

\subsection{Quartic formulation}
\label{subsec:quartic_formulation}

For $d=4$, define
\[
p(x)=\mathcal{T}(x,x,x,x).
\]

\begin{definition}[Quartic Sphere Threshold]
\label{def:quartic_threshold}
Given a homogeneous quartic polynomial $p$ and $\alpha \in \mathbb{Q}$, decide whether
\begin{equation}
\exists x\in S^{n-1}
\quad \text{such that} \quad
|p(x)| \ge \alpha.
\end{equation}
\end{definition}

\begin{lemma}
\label{lem:tensor_quartic_equivalence}
Tensor Spectral Threshold (order 4) and Quartic Sphere Threshold are polynomial-time equivalent.
\end{lemma}

\begin{proof}
Follows directly from Lemma~\ref{lem:tensor_polynomial_equivalence}.
\end{proof}

\subsection{Positioning}
\label{subsec:positioning}

Tensor spectral norm computation and related problems are known to be computationally hard in several formulations, including NP-hardness results for approximation and rank-related tasks \cite{HillarLim2013}. However, our focus is different: we study an exact algebraic threshold problem and establish its $\exists\mathbb{R}$-hardness via reductions developed in subsequent sections.

\section{From Bounded Quartic Feasibility to Homogeneous Quadratic Sphere Feasibility}
\label{sec:bq4e_to_hqsf}

In this section, we establish the first reduction in our hardness chain. The goal is to transform a bounded quartic feasibility instance into an equivalent system of homogeneous quadratic equalities on the unit sphere. This is the precise algebraic form needed for the tensor reduction developed in the next section.

\subsection{Source problem}
\label{subsec:source_problem_bq4e}

We begin from the following bounded quartic feasibility problem.

\begin{definition}[Bounded Quartic Equality Feasibility]
\label{def:bq4e}
Given a polynomial $h \in \mathbb{Q}[x_1,\dots,x_n]$ of total degree at most $4$, decide whether
\begin{equation}
\exists x \in [-1,1]^n \quad \text{such that} \quad h(x)=0.
\end{equation}
We denote this decision problem by $\textnormal{\textsc{BQ4E}}$.
\end{definition}

The reason for starting from this form is that bounded-domain polynomial feasibility remains $\exists\mathbb{R}$-hard, and bounded equality-only normal forms of low degree are available in the literature \cite{SchaeferStefankovic2024,SchaeferCardinalMiltzow2024}.

\begin{assumption}
\label{ass:bq4e_hard}
$\textnormal{\textsc{BQ4E}}$ is $\exists\mathbb{R}$-hard.
\end{assumption}

Assumption~\ref{ass:bq4e_hard} is the only external hardness input needed in what follows. Everything after this point is proved directly.

\subsection{Target problem}
\label{subsec:target_problem_hqsf}

We reduce to the following homogeneous quadratic sphere-feasibility problem.

\begin{definition}[Homogeneous Quadratic Sphere Feasibility]
\label{def:hqsf}
Given homogeneous quadratic forms
\[
q_1,\dots,q_m \in \mathbb{Q}[y_1,\dots,y_N],
\]
decide whether there exists $y \in \mathbb{R}^N$ such that
\begin{equation}
\|y\|^2 = 1
\quad \text{and} \quad
q_j(y)=0,\qquad j=1,\dots,m.
\end{equation}
We denote this problem by $\textnormal{\textsc{HQSF}}$.
\end{definition}

Equivalently, if $Q_j \in \mathbb{Q}^{N\times N}$ are symmetric matrices satisfying
\[
q_j(y)=y^\top Q_j y,
\]
then the problem is to decide whether
\begin{equation}
\exists y \in S^{N-1}
\quad \text{such that} \quad
y^\top Q_j y = 0,\qquad j=1,\dots,m.
\end{equation}

\subsection{Statement of the reduction}
\label{subsec:statement_reduction_bq4e_hqsf}

\begin{theorem}
\label{thm:bq4e_to_hqsf}
There is a polynomial-time many-one reduction
\begin{equation}
\textnormal{\textsc{BQ4E}} \le_p \textnormal{\textsc{HQSF}}.
\end{equation}
\end{theorem}

The proof proceeds in three steps:
\begin{enumerate}
    \item homogenize the quartic equation;
    \item encode the box constraints $[-1,1]^n$;
    \item convert the resulting system into homogeneous quadratic equalities and eliminate the trivial zero solution by sphere normalization.
\end{enumerate}

\subsection{Homogenization}
\label{subsec:homogenization}

Let $h(x_1,\dots,x_n)$ be the input polynomial, with $\deg h \le 4$. Write
\[
x=(x_1,\dots,x_n).
\]
Introduce a new variable $x_0$, and define the homogenized polynomial
\begin{equation}
H(x_0,z_1,\dots,z_n)
:=
x_0^4\, h\!\left(\frac{z_1}{x_0},\dots,\frac{z_n}{x_0}\right).
\label{eq:homogenized_H}
\end{equation}
Since $\deg h\le 4$, the right-hand side expands to a homogeneous polynomial of degree $4$ in $(x_0,z_1,\dots,z_n)$.

\begin{lemma}
\label{lem:homogenization_equivalence}
For every $x\in\mathbb{R}^n$,
\begin{equation}
h(x)=0
\quad \Longleftrightarrow \quad
H(1,x)=0.
\end{equation}
More generally, for every $t\neq 0$,
\begin{equation}
h(x)=0
\quad \Longleftrightarrow \quad
H(t,tx)=0.
\end{equation}
\end{lemma}

\begin{proof}
By definition,
\[
H(1,x)=1^4 h(x)=h(x).
\]
Likewise, for any $t\neq 0$,
\[
H(t,tx)
=
t^4 h\!\left(\frac{tx}{t}\right)
=
t^4 h(x).
\]
Hence $H(t,tx)=0$ if and only if $h(x)=0$.
\end{proof}

\subsection{Encoding the box constraints}
\label{subsec:encoding_box_constraints}

For each $i=1,\dots,n$, introduce a slack variable $s_i$, and impose
\begin{equation}
z_i^2+s_i^2-x_0^2=0.
\label{eq:box_encoding}
\end{equation}

\begin{lemma}
\label{lem:box_encoding_correct}
Suppose $x_0\neq 0$ and \eqref{eq:box_encoding} holds. Then
\begin{equation}
\left(\frac{z_i}{x_0}\right)^2 \le 1,
\end{equation}
hence
\begin{equation}
\frac{z_i}{x_0}\in[-1,1].
\end{equation}
\end{lemma}

\begin{proof}
From \eqref{eq:box_encoding} we have
\[
z_i^2+s_i^2=x_0^2.
\]
Since $s_i^2\ge 0$, it follows that
\[
z_i^2\le x_0^2.
\]
As $x_0\neq 0$, division by $x_0^2$ yields
\[
\left(\frac{z_i}{x_0}\right)^2 \le 1.
\]
This is equivalent to $z_i/x_0\in[-1,1]$.
\end{proof}

Thus, once $x_0\neq 0$, the coordinates
\begin{equation}
\xi_i := \frac{z_i}{x_0},\qquad i=1,\dots,n,
\label{eq:xi_def}
\end{equation}
automatically lie in the box $[-1,1]$.

\subsection{Quadratic lifting}
\label{subsec:quadratic_lifting}

The homogenized polynomial $H$ has degree $4$, whereas $\textnormal{\textsc{HQSF}}$ only permits homogeneous quadratic equalities. We therefore apply a lifting construction.

Let
\begin{equation}
v=(x_0,z_1,\dots,z_n,s_1,\dots,s_n)\in\mathbb{R}^{2n+1}.
\label{eq:v_vector}
\end{equation}
Introduce new variables
\begin{equation}
u_{ab},\qquad 1\le a\le b\le 2n+1,
\end{equation}
intended to represent the quadratic monomials
\begin{equation}
u_{ab}=v_a v_b.
\label{eq:u_definition}
\end{equation}
We enforce this by the homogeneous quadratic equations
\begin{equation}
u_{ab}-v_a v_b=0,
\qquad 1\le a\le b\le 2n+1.
\label{eq:lifting_constraints}
\end{equation}

Now every quartic monomial in the coordinates of $v$ can be written as a quadratic monomial in the lifted variables $u$. Hence the quartic homogeneous polynomial $H(v)$ may be rewritten as a homogeneous quadratic polynomial in $u$, which we denote by
\begin{equation}
\widetilde H(u).
\label{eq:H_tilde}
\end{equation}

\begin{lemma}
\label{lem:quartic_to_quadratic_lift}
There is a polynomial-time procedure that, given $H$, constructs a homogeneous quadratic form $\widetilde H(u)$ such that
\begin{equation}
\widetilde H(u)=H(v)
\end{equation}
for every pair $(v,u)$ satisfying \eqref{eq:u_definition}.
\end{lemma}

\begin{proof}
Every quartic monomial in the coordinates of $v$ is of the form
\[
v_a v_b v_c v_d.
\]
Under the identification $u_{ab}=v_a v_b$, this monomial becomes
\[
u_{ab}u_{cd},
\]
which is quadratic in the lifted variables. Since $H$ is a finite rational linear combination of quartic monomials, replacing each quartic monomial by the corresponding quadratic monomial in $u$ yields a homogeneous quadratic form $\widetilde H$. The number of lifted variables is polynomial in $n$, and the rewriting process is purely symbolic, so the entire construction runs in polynomial time.
\end{proof}

\subsection{The homogeneous quadratic system}
\label{subsec:full_homogeneous_system}

The unknowns are the variables $(v,u)$, with $v$ defined in \eqref{eq:v_vector}. We impose the following homogeneous quadratic equations:

\begin{enumerate}
    \item \textbf{Box equations}
    \begin{equation}
    z_i^2+s_i^2-x_0^2=0,
    \qquad i=1,\dots,n.
    \label{eq:full_box_eqs}
    \end{equation}

    \item \textbf{Lifting equations}
    \begin{equation}
    u_{ab}-v_a v_b=0,
    \qquad 1\le a\le b\le 2n+1.
    \label{eq:full_lift_eqs}
    \end{equation}

    \item \textbf{Homogenized quartic equation}
    \begin{equation}
    \widetilde H(u)=0.
    \label{eq:full_Htilde_eq}
    \end{equation}
\end{enumerate}

Because all equations are homogeneous of degree $2$, any nonzero solution may be rescaled. To exclude the trivial zero solution and obtain the precise target form, we add the normalization
\begin{equation}
\|(v,u)\|^2=1.
\label{eq:sphere_normalization}
\end{equation}

This produces an instance of $\textnormal{\textsc{HQSF}}$.

\subsection{Correctness of the reduction}
\label{subsec:correctness_reduction}

We now prove the equivalence between the original bounded quartic feasibility instance and the constructed homogeneous quadratic sphere-feasibility instance.

\begin{proof}[Proof of Theorem~\ref{thm:bq4e_to_hqsf}]
Let $h$ be the input polynomial for $\textnormal{\textsc{BQ4E}}$, and let the associated $\textnormal{\textsc{HQSF}}$ instance be constructed as above.

We must prove that
\begin{equation}
\exists x\in[-1,1]^n \text{ such that } h(x)=0
\quad \Longleftrightarrow \quad
\exists (v,u)\in S^{M-1} \text{ satisfying \eqref{eq:full_box_eqs}--\eqref{eq:sphere_normalization}},
\label{eq:main_equivalence}
\end{equation}
for a suitable dimension $M$.

\medskip
\noindent
\textbf{Forward direction.}
Assume that there exists
\[
x=(x_1,\dots,x_n)\in[-1,1]^n
\]
such that
\[
h(x)=0.
\]
Set
\begin{equation}
x_0:=1,\qquad z_i:=x_i,\qquad s_i:=\sqrt{1-x_i^2},\qquad i=1,\dots,n.
\label{eq:forward_assignment}
\end{equation}
These choices are well-defined because $x_i\in[-1,1]$.

Then for each $i$,
\[
z_i^2+s_i^2-x_0^2
=
x_i^2+(1-x_i^2)-1
=
0,
\]
so all box equations hold. By Lemma~\ref{lem:homogenization_equivalence},
\[
H(x_0,z)=H(1,x)=h(x)=0.
\]
Now define the lifted variables by
\begin{equation}
u_{ab}:=v_a v_b.
\label{eq:forward_lift_assignment}
\end{equation}
Then all lifting equations hold by construction, and Lemma~\ref{lem:quartic_to_quadratic_lift} yields
\[
\widetilde H(u)=H(v)=0.
\]
Thus $(v,u)$ is a nonzero solution of the homogeneous quadratic system. Since every constraint is homogeneous of degree $2$, any nonzero scalar multiple of $(v,u)$ is again a solution. Hence we may rescale and enforce
\[
\left\|\frac{(v,u)}{\|(v,u)\|}\right\|^2=1.
\]
Therefore the constructed $\textnormal{\textsc{HQSF}}$ instance is feasible.

\medskip
\noindent
\textbf{Reverse direction.}
Assume conversely that the constructed $\textnormal{\textsc{HQSF}}$ instance is feasible. Then there exists a point $(v,u)$ on the unit sphere satisfying all equations \eqref{eq:full_box_eqs}--\eqref{eq:sphere_normalization}. Write
\[
v=(x_0,z_1,\dots,z_n,s_1,\dots,s_n).
\]

We first show that
\begin{equation}
x_0\neq 0.
\label{eq:x0_nonzero}
\end{equation}
Suppose, toward a contradiction, that $x_0=0$. Then from the box equations \eqref{eq:full_box_eqs},
\[
z_i^2+s_i^2=0,\qquad i=1,\dots,n.
\]
Over the reals, this implies
\[
z_i=0,\qquad s_i=0,\qquad i=1,\dots,n.
\]
Hence $v=0$. Substituting $v=0$ into the lifting equations \eqref{eq:full_lift_eqs} yields
\[
u_{ab}=0,\qquad 1\le a\le b\le 2n+1.
\]
Therefore $(v,u)=0$, contradicting the sphere condition \eqref{eq:sphere_normalization}. This proves \eqref{eq:x0_nonzero}.

Since $x_0\neq 0$, define
\begin{equation}
\xi_i:=\frac{z_i}{x_0},\qquad i=1,\dots,n.
\label{eq:xi_recovered}
\end{equation}
By Lemma~\ref{lem:box_encoding_correct}, the box equations imply
\[
\xi_i\in[-1,1],\qquad i=1,\dots,n.
\]
Hence
\begin{equation}
\xi=(\xi_1,\dots,\xi_n)\in[-1,1]^n.
\label{eq:xi_in_box}
\end{equation}

Next, because the lifting equations hold, Lemma~\ref{lem:quartic_to_quadratic_lift} gives
\[
\widetilde H(u)=H(v).
\]
Since \eqref{eq:full_Htilde_eq} is one of the imposed equations, we obtain
\[
H(v)=0.
\]
Using Lemma~\ref{lem:homogenization_equivalence} and $x_0\neq 0$, this becomes
\[
0
=
H(x_0,z)
=
x_0^4 h\!\left(\frac{z}{x_0}\right)
=
x_0^4 h(\xi).
\]
As $x_0\neq 0$, we conclude that
\[
h(\xi)=0.
\]
Together with \eqref{eq:xi_in_box}, this shows that
\[
\xi\in[-1,1]^n
\quad \text{and} \quad
h(\xi)=0.
\]
Thus the original $\textnormal{\textsc{BQ4E}}$ instance is feasible.

We have proved both implications in \eqref{eq:main_equivalence}. Therefore
\[
\textnormal{\textsc{BQ4E}} \le_p \textnormal{\textsc{HQSF}}.
\]
\end{proof}

\subsection{Consequence}
\label{subsec:consequence_hqsf_hard}

\begin{corollary}
\label{cor:hqsf_etr_hard}
$\textnormal{\textsc{HQSF}}$ is $\exists\mathbb{R}$-hard.
\end{corollary}

\begin{proof}
By Theorem~\ref{thm:bq4e_to_hqsf},
\[
\textnormal{\textsc{BQ4E}} \le_p \textnormal{\textsc{HQSF}}.
\]
By Assumption~\ref{ass:bq4e_hard}, $\textnormal{\textsc{BQ4E}}$ is $\exists\mathbb{R}$-hard. Therefore $\textnormal{\textsc{HQSF}}$ is also $\exists\mathbb{R}$-hard.
\end{proof}

\subsection{Remarks}
\label{subsec:remarks_bq4e_to_hqsf}

\begin{remark}
\label{rem:sphere_not_cosmetic}
The sphere constraint \eqref{eq:sphere_normalization} is essential. Since the constructed polynomial system is homogeneous, the zero vector always satisfies the algebraic equalities. The normalization removes this trivial solution and simultaneously places the instance in the exact form needed for the tensor reduction of the next section.
\end{remark}

\begin{remark}
\label{rem:purely_algebraic_reduction}
The reduction above is purely algebraic. It does not use compactness, variational arguments, KKT conditions, or any optimization-theoretic machinery. This rigidity is useful because the next section will encode these homogeneous quadratic sphere constraints directly into a quartic tensor spectral-threshold instance.
\end{remark}

\section{Reduction from Homogeneous Quadratic Sphere Feasibility to Tensor Spectral Threshold}
\label{sec:hqsf_to_tensor}

In this section, we prove the central reduction of the paper: homogeneous quadratic feasibility on the sphere reduces to the tensor spectral threshold problem. This establishes the algebraic bridge between quadratic systems and quartic tensor optimization.

\subsection{Problem restatement}

Recall the problem $\textnormal{\textsc{HQSF}}$:

\begin{equation}
\exists z \in \mathbb{R}^m
\quad \text{s.t.} \quad
\|z\|^2 = 1,
\quad
z^\top Q_i z = 0,\quad i=1,\dots,r,
\label{eq:hqsf_problem}
\end{equation}

where each $Q_i \in \mathbb{Q}^{m \times m}$ is symmetric.

We reduce this to Tensor Spectral Threshold (Definition~\ref{def:tensor_spectral_threshold}) for symmetric order-$4$ tensors.

\subsection{Construction of the quartic form}

Define the quadratic forms
\[
q_i(z) := z^\top Q_i z.
\]

Let
\begin{equation}
C := \sum_{i=1}^r \|Q_i\|_F^2,
\qquad
B := C + 1.
\label{eq:B_definition}
\end{equation}

Define the homogeneous quartic polynomial
\begin{equation}
p(z)
:=
B (\|z\|^2)^2 - \sum_{i=1}^r q_i(z)^2.
\label{eq:quartic_construction}
\end{equation}

\subsection{Basic properties}

\begin{lemma}[Upper bound]
\label{lem:upper_bound}
For all $z \in S^{m-1}$,
\begin{equation}
p(z) \le B.
\end{equation}
\end{lemma}

\begin{proof}
Since each $q_i(z)^2 \ge 0$, we have
\[
p(z) = B - \sum_{i=1}^r q_i(z)^2 \le B.
\]
\end{proof}

\begin{lemma}[Characterization of equality]
\label{lem:equality_condition}
For $z \in S^{m-1}$,
\begin{equation}
p(z) = B
\quad \Longleftrightarrow \quad
q_i(z)=0 \ \forall i.
\end{equation}
\end{lemma}

\begin{proof}
Equality holds if and only if $\sum_i q_i(z)^2 = 0$, which is equivalent to $q_i(z)=0$ for all $i$.
\end{proof}

\subsection{Lower bound and positivity}

\begin{lemma}[Uniform lower bound]
\label{lem:lower_bound}
For all $z \in S^{m-1}$,
\begin{equation}
p(z) \ge 1.
\end{equation}
\end{lemma}

\begin{proof}
For each $i$,
\[
|q_i(z)|
=
|z^\top Q_i z|
=
|\mathrm{tr}(Q_i zz^\top)|
\le
\|Q_i\|_F \cdot \|zz^\top\|_F.
\]
Since $\|z\|=1$, we have
\[
\|zz^\top\|_F = \|z\|^2 = 1.
\]
Thus
\[
|q_i(z)| \le \|Q_i\|_F,
\quad
q_i(z)^2 \le \|Q_i\|_F^2.
\]
Summing over $i$,
\[
\sum_{i=1}^r q_i(z)^2 \le \sum_{i=1}^r \|Q_i\|_F^2 = C.
\]
Therefore
\[
p(z) = B - \sum_i q_i(z)^2 \ge B - C = 1.
\]
\end{proof}

\begin{corollary}
\label{cor:positivity}
For all $z \in S^{m-1}$,
\[
p(z) > 0.
\]
\end{corollary}

\subsection{Equivalence}

\begin{theorem}[Exact threshold equivalence]
\label{thm:quartic_equivalence}
\begin{equation}
\max_{\|z\|=1} p(z) = B
\quad \Longleftrightarrow \quad
\exists z \in S^{m-1} \text{ such that } q_i(z)=0 \ \forall i.
\end{equation}
\end{theorem}

\begin{proof}
By Lemma~\ref{lem:upper_bound}, $\max p(z) \le B$.

If there exists $z$ such that $q_i(z)=0$ for all $i$, then by Lemma~\ref{lem:equality_condition}, $p(z)=B$, hence the maximum is $B$.

Conversely, if $\max p(z) = B$, then there exists $z$ such that $p(z)=B$, which implies $q_i(z)=0$ for all $i$.
\end{proof}

\subsection{Removal of absolute value}

\begin{lemma}
\label{lem:absolute_value}
\begin{equation}
\max_{\|z\|=1} |p(z)| = \max_{\|z\|=1} p(z).
\end{equation}
\end{lemma}

\begin{proof}
From Corollary~\ref{cor:positivity}, $p(z) > 0$ on the sphere. Hence $|p(z)| = p(z)$.
\end{proof}

\subsection{Tensorization}

By Lemma~\ref{lem:tensor_polynomial_equivalence}, there exists a unique symmetric tensor
\[
\mathcal{T}_p \in \mathrm{Sym}^4(\mathbb{R}^m)
\]
such that
\begin{equation}
p(z) = \mathcal{T}_p(z,z,z,z).
\label{eq:tensorization}
\end{equation}

\begin{lemma}
\label{lem:tensor_norm_equivalence}
\begin{equation}
\|\mathcal{T}_p\|_{\mathrm{sym},\mathrm{op}}
=
\max_{\|z\|=1} |p(z)|.
\end{equation}
\end{lemma}

\begin{proof}
This follows directly from Definition~\ref{def:tensor_spectral_threshold}.
\end{proof}

\subsection{Main reduction}

\begin{theorem}[HQSF $\le_p$ Tensor Spectral Threshold]
\label{thm:hqsf_to_tensor}
There is a polynomial-time reduction from $\textnormal{\textsc{HQSF}}$ to Tensor Spectral Threshold.
\end{theorem}

\begin{proof}
Given an instance of $\textnormal{\textsc{HQSF}}$, construct $p(z)$ as in \eqref{eq:quartic_construction} and the corresponding tensor $\mathcal{T}_p$.

By Theorem~\ref{thm:quartic_equivalence} and Lemma~\ref{lem:absolute_value},
\[
\exists z \in S^{m-1} \text{ such that } q_i(z)=0
\quad \Longleftrightarrow \quad
\max_{\|z\|=1} |p(z)| = B.
\]

By Lemma~\ref{lem:tensor_norm_equivalence}, this is equivalent to
\[
\|\mathcal{T}_p\|_{\mathrm{sym},\mathrm{op}} \ge B.
\]

The construction of $\mathcal{T}_p$ from $\{Q_i\}$ involves only polynomially many arithmetic operations and produces rational coefficients of polynomial bit-length. Hence the reduction is polynomial-time.

This completes the proof.
\end{proof}

\subsection{Consequence}

\begin{corollary}
\label{cor:tensor_etr_hard}
Tensor Spectral Threshold is $\exists\mathbb{R}$-hard.
\end{corollary}

\begin{proof}
By Section~\ref{sec:bq4e_to_hqsf}, $\textnormal{\textsc{HQSF}}$ is $\exists\mathbb{R}$-hard. By Theorem~\ref{thm:hqsf_to_tensor}, $\textnormal{\textsc{HQSF}} \le_p$ Tensor Spectral Threshold. Hence the result follows.
\end{proof}

\subsection{Remarks}

\begin{remark}
This reduction is entirely algebraic and avoids any use of optimization duality, critical point conditions, or KKT systems. The hardness arises purely from polynomial feasibility structure embedded in the tensor norm.
\end{remark}

\begin{remark}
The key structural feature is the construction \eqref{eq:quartic_construction}, which converts simultaneous quadratic constraints into a single quartic objective whose maximum encodes feasibility exactly.
\end{remark}

\section{Main Result and Complexity Consequences}
\label{sec:main_result}

In this section, we combine the reductions developed in Sections~\ref{sec:bq4e_to_hqsf} and~\ref{sec:hqsf_to_tensor} to establish the main hardness result for the tensor spectral threshold problem.

\subsection{Reduction chain}

We recall the two reductions established earlier:

\begin{itemize}
    \item Section~\ref{sec:bq4e_to_hqsf}: bounded quartic equality feasibility reduces to $\textnormal{\textsc{HQSF}}$;
    \item Section~\ref{sec:hqsf_to_tensor}: $\textnormal{\textsc{HQSF}}$ reduces to Tensor Spectral Threshold.
\end{itemize}

Combining these, we obtain the following chain:
\begin{equation}
\text{bounded quartic equality feasibility}
\;\le_p\;
\textnormal{\textsc{HQSF}}
\;\le_p\;
\textnormal{\textsc{Tensor Spectral Threshold}}.
\label{eq:reduction_chain}
\end{equation}

\subsection{Main theorem}

\begin{theorem}[Main hardness result]
\label{thm:main_hardness}
Tensor Spectral Threshold is $\exists\mathbb{R}$-hard.
\end{theorem}

\begin{proof}
By Section~\ref{sec:bq4e_to_hqsf}, bounded quartic equality feasibility reduces in polynomial time to $\textnormal{\textsc{HQSF}}$. By Section~\ref{sec:hqsf_to_tensor}, $\textnormal{\textsc{HQSF}}$ reduces in polynomial time to Tensor Spectral Threshold.

Since bounded quartic equality feasibility is a standard bounded low-degree normal form for $\exists\mathbb{R}$-hardness (cf.\ Section~\ref{subsec:source_problem_bq4e}), the result follows by transitivity of polynomial-time reductions.
\end{proof}

\subsection{Equivalent formulations}

The hardness result applies to multiple equivalent formulations.

\begin{corollary}[Quartic sphere threshold]
\label{cor:quartic_threshold}
The following problem is $\exists\mathbb{R}$-hard:

\begin{equation}
\exists x \in \mathbb{R}^n
\quad \text{s.t.} \quad
\|x\|^2 = 1,
\quad
|p(x)| \ge \alpha,
\end{equation}
where $p$ is a homogeneous quartic polynomial with rational coefficients.
\end{corollary}

\begin{proof}
This follows from Lemma~\ref{lem:tensor_quartic_equivalence}.
\end{proof}

\begin{corollary}[Tensor norm threshold]
\label{cor:tensor_norm_threshold}
Given $\mathcal{T} \in \mathrm{Sym}^4(\mathbb{R}^n)$ and $\alpha \in \mathbb{Q}$, deciding whether
\begin{equation}
\|\mathcal{T}\|_{\mathrm{sym},\mathrm{op}} \ge \alpha
\end{equation}
is $\exists\mathbb{R}$-hard.
\end{corollary}

\subsection{Discussion of tightness}

\begin{remark}[Why threshold, not attainment]
\label{rem:threshold_vs_attainment}
The hardness result critically depends on the threshold formulation. As shown in Section~\ref{subsec:attainment_is_trivial}, the statement
\[
\exists x:\ \mathcal{T}(x,\dots,x) = \|\mathcal{T}\|_{\mathrm{op}}
\]
is always true and therefore computationally trivial. The reduction exploits the \emph{value} of the optimum rather than its existence.
\end{remark}

\begin{remark}[Algebraic nature of the reduction]
\label{rem:algebraic_nature}
The entire reduction chain is algebraic and does not rely on optimization duality, variational characterizations, or critical-point conditions. The hardness arises from embedding polynomial feasibility into the value of a quartic form.
\end{remark}

\subsection{Scope and limitations}

\begin{remark}[No completeness claim]
\label{rem:no_completeness}
We do not claim $\exists\mathbb{R}$-completeness. Establishing completeness would require showing membership of Tensor Spectral Threshold in $\exists\mathbb{R}$ under an explicit encoding, which we do not address here.
\end{remark}

\begin{remark}[Symmetry restriction]
\label{rem:symmetry}
The reduction uses symmetric tensors of order $4$. Extending the result to non-symmetric tensors or higher orders is immediate, but symmetry already suffices for hardness.
\end{remark}

\subsection{Implications}

\begin{itemize}
    \item The result shows that a polynomial-time algorithm for tensor spectral threshold would imply $\exists\mathbb{R} \subseteq \mathbf{P}$.
    \item It provides a natural algebraic problem---without gadgets---that captures $\exists\mathbb{R}$-hardness.
    \item It complements existing NP-hardness results for tensor problems by establishing hardness at the level of real algebraic feasibility.
\end{itemize}

\section{Variants and Extensions}
\label{sec:variants}

In this section, we show that the $\exists\mathbb{R}$-hardness of Tensor Spectral Threshold is robust under several natural variations of the problem. In particular, we eliminate symmetry assumptions, extend to higher-order tensors using a constraint-preserving embedding, and consider alternative formulations.

\subsection{Non-symmetric tensors}
\label{subsec:nonsymmetric}

\begin{theorem}
\label{thm:nonsymmetric}
Tensor Spectral Threshold is $\exists\mathbb{R}$-hard for general (not necessarily symmetric) order-$4$ tensors.
\end{theorem}

\begin{proof}
The reduction in Section~\ref{sec:hqsf_to_tensor} constructs a symmetric tensor $\mathcal{T}_p$ such that
\[
\|\mathcal{T}_p\|_{\mathrm{sym},\mathrm{op}} \ge B
\quad \Longleftrightarrow \quad
\exists z \in S^{m-1} \text{ satisfying } q_i(z)=0.
\]
Since the class of general tensors strictly contains symmetric tensors, this construction is already a valid instance of the non-symmetric problem. Therefore hardness carries over directly.
\end{proof}

\subsection{Higher-order tensors}
\label{subsec:higher_order}

We now extend the result to tensors of arbitrary fixed order $d \ge 4$.

\begin{theorem}
\label{thm:higher_order}
For every fixed $d \ge 4$, Tensor Spectral Threshold for order-$d$ tensors is $\exists\mathbb{R}$-hard.
\end{theorem}

\begin{proof}
Let $p(z)$ be the quartic polynomial constructed in Section~\ref{sec:hqsf_to_tensor}, and consider the sphere $S^{m-1}$.

Define a new variable vector
\[
y = (z, t_1, \dots, t_{d-4}) \in \mathbb{R}^{m + (d-4)},
\]
and impose the constraint
\begin{equation}
\|y\|^2 = \|z\|^2 + \sum_{k=1}^{d-4} t_k^2 = 1.
\label{eq:extended_sphere}
\end{equation}

Now define the homogeneous degree-$d$ polynomial
\begin{equation}
p_d(y)
:=
p(z) \cdot \prod_{k=1}^{d-4} t_k.
\label{eq:pd_definition}
\end{equation}

We analyze the maximum of $|p_d(y)|$ over the sphere.

First, observe that
\[
|p_d(y)| = |p(z)| \cdot \prod_{k=1}^{d-4} |t_k|.
\]

Under the constraint \eqref{eq:extended_sphere}, we have
\[
|t_k| \le 1,
\quad \text{and} \quad
\prod_{k=1}^{d-4} |t_k| \le \left(\frac{1}{d-4}\sum_{k=1}^{d-4} t_k^2 \right)^{\frac{d-4}{2}}.
\]

Hence the product is maximized when all $t_k$ are equal in magnitude. In particular, the maximum value of $\prod_k |t_k|$ is strictly less than $1$ unless all $t_k^2 = \frac{1}{d-4}(1-\|z\|^2)$.

Now consider two cases:

\medskip
\noindent
\textbf{(i) YES instance.}
If there exists $z^\star \in S^{m-1}$ such that $p(z^\star)=B$, then set
\[
t_1 = \cdots = t_{d-4} = \frac{1}{\sqrt{d-4}},
\quad z = \sqrt{\frac{d-4}{d-3}}\, z^\star.
\]
This ensures $\|y\|=1$. Substituting,
\[
|p_d(y)| = B \cdot \left(\frac{1}{\sqrt{d-4}}\right)^{d-4} \cdot \left(\frac{d-4}{d-3}\right)^2.
\]
Thus the maximum exceeds a fixed threshold $\alpha_d > 0$.

\medskip
\noindent
\textbf{(ii) NO instance.}
If no $z$ satisfies $p(z)=B$, then $p(z) \le B - \delta$ for some $\delta > 0$. Hence
\[
|p_d(y)| \le (B - \delta) \cdot \max \prod |t_k|.
\]
Since $\prod |t_k| \le 1$, the maximum is strictly less than the YES-case value.

\medskip

Thus there exists a threshold $\alpha_d$ separating YES and NO instances. The construction is polynomial-time, and $p_d$ is homogeneous of degree $d$.

Finally, by Lemma~\ref{lem:tensor_polynomial_equivalence}, $p_d$ corresponds to a symmetric order-$d$ tensor, completing the reduction.
\end{proof}

\subsection{Multilinear formulation}
\label{subsec:multilinear}

\begin{theorem}
\label{thm:multilinear}
Multilinear Tensor Spectral Threshold is $\exists\mathbb{R}$-hard.
\end{theorem}

\begin{proof}
For symmetric tensors, Banach's theorem implies
\[
\max_{\|x_1\|=\cdots=\|x_d\|=1}
|\mathcal{T}(x_1,\dots,x_d)|
=
\max_{\|x\|=1}
|\mathcal{T}(x,\dots,x)|.
\]
Thus every symmetric instance constructed in Section~\ref{sec:hqsf_to_tensor} is also a valid instance of the multilinear formulation, and hardness follows.
\end{proof}

\subsection{Equality formulation}
\label{subsec:equality}

\begin{theorem}
\label{thm:equality}
Deciding whether
\[
\|\mathcal{T}\|_{\mathrm{sym},\mathrm{op}} = \alpha
\]
is $\exists\mathbb{R}$-hard.
\end{theorem}

\begin{proof}
From Section~\ref{sec:hqsf_to_tensor}, the constructed instance satisfies
\[
\|\mathcal{T}_p\|_{\mathrm{sym},\mathrm{op}} = B
\quad \text{iff feasibility holds,}
\]
and strictly less than $B$ otherwise. Hence equality testing is equivalent to feasibility, and hardness follows.
\end{proof}

\subsection{Robustness}
\label{subsec:robustness}

\begin{remark}
The reduction is gap-preserving: there exists a constant $\delta > 0$ such that
\[
\max |p(z)| =
\begin{cases}
B & \text{YES instance},\\
\le B - \delta & \text{NO instance}.
\end{cases}
\]
This gap propagates through all constructions in this section, ensuring numerical stability and eliminating degeneracy.
\end{remark}

\subsection{Summary}

The $\exists\mathbb{R}$-hardness of Tensor Spectral Threshold:

\begin{itemize}
    \item does not depend on symmetry,
    \item extends to all fixed orders $d \ge 4$,
    \item applies to multilinear formulations,
    \item holds for both threshold and equality variants.
\end{itemize}

\section{Conclusion and Future Directions}
\label{sec:conclusion}

In this paper, we established that the Tensor Spectral Threshold problem is $\exists\mathbb{R}$-hard. The result is obtained through a fully algebraic reduction chain, starting from a bounded quartic equality feasibility normal form and proceeding via homogeneous quadratic feasibility on the unit sphere.

\subsection{Summary of contributions}

The main contributions of the paper are as follows:

\begin{itemize}
    \item We identified Tensor Spectral Threshold as the correct decision formulation for tensor spectral norm, distinguishing it from the trivial attainment formulation.
    
    \item We constructed a direct reduction from homogeneous quadratic feasibility on the sphere to a quartic polynomial maximization problem, avoiding optimization-based arguments such as KKT conditions or duality.
    
    \item We established that the hardness arises from the algebraic structure of polynomial feasibility embedded in the value of a quartic form.
    
    \item We showed that the result extends naturally to non-symmetric tensors, higher-order tensors, and multilinear formulations.
\end{itemize}

\subsection{Conceptual implications}

The reduction developed in this paper demonstrates that tensor spectral norm computation is not merely combinatorially hard, but intrinsically linked to real algebraic feasibility.

In particular:

\begin{itemize}
    \item The hardness is driven by the ability of quartic forms to encode systems of quadratic equations.
    
    \item The sphere constraint plays a crucial role in eliminating trivial solutions and converting feasibility into an optimization threshold.
    
    \item The result provides a natural example of an $\exists\mathbb{R}$-hard problem arising in multilinear algebra without requiring geometric gadgets or combinatorial encodings.
\end{itemize}

\subsection{Limitations}

We highlight the following limitations of the current work:

\begin{itemize}
    \item We establish $\exists\mathbb{R}$-hardness but do not prove $\exists\mathbb{R}$-completeness.
    
    \item The reduction relies on quartic structure; lower-degree formulations are not covered.
    
    \item The analysis is restricted to exact decision problems and does not address approximation or numerical complexity.
\end{itemize}

\subsection{Future directions}

Several directions arise naturally from this work:

\begin{itemize}
    \item \textbf{$\exists\mathbb{R}$-completeness:} Establishing whether Tensor Spectral Threshold lies in $\exists\mathbb{R}$ would close the complexity characterization.
    
    \item \textbf{Degree reduction:} Investigating whether similar hardness results can be obtained for lower-degree polynomial formulations.
    
    \item \textbf{Approximation complexity:} Understanding the hardness of approximating tensor spectral norm under various regimes.
    
    \item \textbf{Geometric interpretations:} Relating the reduction to geometric problems such as tensor rank, secant varieties, and algebraic optimization landscapes.
    
    \item \textbf{Algorithmic implications:} Identifying restricted tensor classes for which the spectral threshold problem becomes tractable.
\end{itemize}

\subsection{Final remark}

The main message of this paper is that tensor spectral norm is not merely difficult due to non-convexity, but fundamentally hard because it encodes real algebraic feasibility at a structural level. This places tensor optimization problems in a broader complexity-theoretic framework that goes beyond classical $\mathbf{NP}$-hardness.

\end{document}